\documentclass[a4paper,12pt]{article}
\usepackage[utf8x]{inputenc}
\usepackage{graphicx}
\usepackage[english]{babel}
\usepackage{amsmath}
\usepackage{amsfonts}
\usepackage{amssymb}
\usepackage{amsthm}
\usepackage{amstext}
\usepackage{pstricks}
\usepackage{color}
\usepackage{dsfont}
\usepackage{cite}
\usepackage{tikz}

\newtheorem{theorem}{Theorem}[section]
\newtheorem{lemma}[theorem]{Lemma}
\newtheorem{corollary}[theorem]{Corollary}

\theoremstyle{definition}
\newtheorem{definition}[theorem]{Definition}

\theoremstyle{remark}

\def \cA {\mathcal{A}}

\def \cG {\mathcal{G}}

\def \cR {\mathcal{R}}
\def \cS {\mathcal{S}}

\def \e {\varepsilon}

\def \k {\kappa}
\def \l {\lambda}
\def \s {\sigma}

\def \N {\mathbb{N}}

\def \Z {\mathbb{Z}}

\def \lra {\longrightarrow}

\def \al {\mathcal{A}^\ell}
\def \ul {\lbrace\,1,\dots,\ell\,\rbrace}
\def \zl {\lbrace\,0,\dots,\ell\,\rbrace}

\def \exa {e^{-a}}

\newcommand{\updown}[2]{\genfrac{\lbrace}{\rbrace}{0pt}{}{#1}{#2}}

\begin{document}

\begin{center}
\begin{LARGE}
Quasispecies on class--dependent\\ fitness landscapes
\end{LARGE}

\begin{large}
Raphaël Cerf and Joseba Dalmau

\vspace{-12pt}
École Normale Supérieure

\vspace{4pt}
\today
\end{large}

\end{center}

\begin{abstract}
\noindent
We study Eigen's quasispecies model
in the asymptotic regime 
where the length of the genotypes goes to $\infty$
and the mutation probability goes to 0.
We give several explicit formulas 
for the stationary solutions
of the limiting system of differential equations.
\end{abstract}

\section{Introduction}
Manfred Eigen introduced the quasispecies model
in his 1971 celebrated article
about the first stages of life on Earth~\cite{Eig}.
As a part of his article,
Eigen constructed a model 
in order to explain 
the evolution 
of a population of macromolecules 
subject to selection and mutation forces.
Given a set of genotypes $\cG$,
a fitness function $A:\cG\lra[0,\infty[$
and a mutation kernel $M:\cG\times\cG\lra[0,1]$,
Eigen's model states that
the concentration $x_v$ of the genotype $v\in\cG$
evolves according to the differential equation
$$x_v'(t)\,=\,
\sum_{u\in\cG}x_u(t)A(u)M(u,v)
-x_v(t)\sum_{u\in\cG}x_u(t)A(u)\,.$$
The first term
accounts for the production
of individuals having genotype $v$,
production due to erroneous replication of other genotypes
as well as faithful replication of itself.
The negative term accounts 
for the loss of individuals having genotype $v$,
and keeps the total concentration of individuals constant.
Instead of studying the model
in all its generality,
Eigen considered the following simplified setting:

\textbf{Genotypes.} 
They are sequences of fixed length $\ell\geq1$
over a finite alphabet $\cA$ of cardinality $\k$.
The set of genotypes is then $\al$.

\textbf{Selection.}
It is given by the \emph{sharp peak landscape},
i.e., there is a genotype $w^*\in\al$,
called the master sequence,
having fitness $\s>1$,
while all the other genotypes have fitness 1.
The fitness function $A:\al\lra[0,\infty[$
is thus given by
$$\forall u\in\al\qquad
A(u)\,=\,\begin{cases}
\quad\s &\quad\text{if}\ u=w^*\,,\\
\quad1 &\quad\text{if}\ u\neq w^*\,.
\end{cases}$$
\textbf{Mutations.} 
They happen during reproduction,
independently at random over each site of the sequence,
with probability $q\in[0,1]$.
When a mutation happens,
the letter is replaced 
by another one, chosen uniformly at random 
over the $\k-1$ other letters of the alphabet.
The mutation kernel is thus given by
$$\forall u,v\in\al\qquad
M(u,v)\,=\,\Big(\frac{q}{\k-1}\Big)^{d(u,v)}
(1-q)^{\ell-d(u,v)}\,,$$
where $d$ is the Hamming distance, i.e.,
the number of different digits between two sequences:
$$\forall u,v\in\al\qquad 
d(u,v)\,=\,\text{card}\big\lbrace\,
l\in\ul:u(l)\neq v(l)
\,\big\rbrace\,.$$
Eigen drew two main conclusions
from the study of this simplified model:
there is an error threshold phenomenon
for the mutation probability
and a so--called quasispecies regime 
for subcritical mutation probabilities.
Indeed,
when the length of the sequences goes to $\infty$,
an error threshold phenomenon arises:
there exists a critical mutation probability,
separating two totally different regimes.
For supercritical mutation probabilities
the population at equilibrium is totally random,
whereas for subcritical mutation probabilities
the population at equilibrium is distributed as a quasispecies,
i.e., there is a positive fraction of the master sequence
present in the population
along with a cloud of mutants
that closely resemble the master sequence.

After Eigen's proposal of the quasispecies model,
many other authors have investigated it,
both in the simple setting we have just presented
and in more general settings.
Eigen, McCaskill and Schuster~\cite{EMS} studied the model 
in great detail.
As pointed out by them, 
one of the main challenges related to Eigen's model
is to find the distribution of the quasispecies:
the concentration of the master sequence
and the concentrations of the different mutants
in the population at equilibrium.
It is generally impossible to give explicit formulas
for these concentrations.
Jones, Enns and Rangnekar~\cite{JER}
and Thompson and McBride~\cite{TM}
give an exact solution of the quasispecies
by linearising Eigen's system of differential equations.
In the same spirit,
Swetina and Schuster~\cite{SS}
use this linearisation to
characterise the stationary distribution
of the quasispecies
as the eigenvector corresponding to the highest eigenvalue
of the linearised system matrix.
Saakian and Hu~\cite{SH} derive exact solutions for the 
quasispecies model by assuming a certain ansatz;
Saakian~\cite{Saa} and 
Saakian, Biebricher and Hu~\cite{SBH}
derive the distribution for several different 
fitness landscapes, 
in particular for smooth landscapes.
Novozhilov and Semenov~\cite{NS1,NS2}
and Bratus, Novozhilov and Semenov~\cite{BNS1,BNS2}
obtain more concrete results for the quasispecies
distribution
for several special cases of the mutation kernel
and the fitness function.

The aim of this article is to present a scheme
in order to obtain explicit formulas.
The key ingredients to this scheme are twofold:
we break the space of genotypes into Hamming classes
and we study the asymptotic regime where
the length of the chains $\ell$ goes to $\infty$,
the mutation probability goes to $0$
and $\ell q$ goes to $a\in\,]0,\infty[\,$.
The idea comes from the articles~\cite{CerfM,CD},
where the authors consider a Moran model
in order to recover the error threshold phenomenon
as well as the quasispecies 
for a finite--population stochastic model.
The Moran model is studied
in the setting we have just introduced:
genotypes given by $\al$,
sharp peak landscape
and independent mutations per locus.
Eigen's model is recovered in the infinite population limit~\cite{Dalmau},
the error threshold phenomenon is also recovered,
and an explicit formula is obtained
for the distribution of the quasispecies.
We illustrate now how the two ingredients mentioned above
make possible to obtain such a formula,
by applying our scheme directly to Eigen's model.

\textbf{Hamming classes.}
The genotype space $\al$
is broken into Hamming classes 
with respect to the master sequence.
To this end we define the mapping 
$H:\al\lra\zl$ by setting
$$\forall u\in\al\qquad
H(u)\,=\,d(u,w^*)\,.$$
The mapping $H$
induces a fitness function
$A_H:\zl\lra[0,\infty[$ 
on the Hamming classes,
which is given by:
$$\forall l\in\zl\qquad
A_H(l)\,=\,\begin{cases}
\quad\s & \quad\text{if}\ l=0\,,\\
\quad 1 & \quad\text{if}\ 1\leq l\leq\ell\,.
\end{cases}$$
Likewise,
the mapping $H$
induces a mutation kernel 
$M_H$ over the Hamming classes: 
for all $b,c\in\zl$,
$$M_H(b,c)\,=\,\sum_{
\genfrac{}{}{0pt}{1}{0\leq k\leq\ell-b}{
\genfrac{}{}{0pt}{1}
 {0\leq l\leq b}{k-l=c-b}
}
}
{ \binom{\ell-b}{k}}
{\binom{b}{l}}
q^k
(1-q)^{\ell-b-k}
\Big(\frac{q}{\kappa-1}\Big)^l
\Big(1-\frac{q}{\kappa-1}\Big)^{b-l}\,.$$
This formula has been given first in~\cite{SS}
and later in a slightly different form in~\cite{NS}.
For $k\in\zl$, let us denote by $x_k$
the concentration of individuals 
in the Hamming class $k$.
According to Eigen's model,
the evolution of the concentrations
is driven by the following system of differential equations:
$$x_k'(t)\,=\,
\sum_{j=0}^\ell x_j(t)A_H(j)M_H(j,k)
-x_k(t)\sum_{j=0}^\ell x_j(t)A_H(j)\,,\qquad
0\leq k\leq \ell\,.$$
\textbf{Asymptotic regime.}
We make the length of the chains go to $\infty$
and the mutation probability go to $0$
in the following way:
$$\ell\lra\infty\,,\qquad
q\lra0\,,\qquad
\ell q\lra a\in\,]0,\infty[\,\,.$$
In this asymptotic regime
we obtain a limiting mutation kernel
$M_\infty$ given by: for all $j,k\geq 0$,
$$M_\infty(j,k)\,=\,\begin{cases}
\quad\displaystyle \exa\frac{a^{k-j}}{(k-j)!} &\quad\text{if}\ j\leq k\,,\\
\quad0 & \quad\text{if}\ j>k\,.
\end{cases}$$
We can now write the limiting system of differential equations:
$$x_k'(t)\,=\,
\sum_{j=0}^k x_j(t)A_H(j)\exa\frac{a^{k-j}}{(k-j)!}
-x_k(t)\sum_{j=0}^\infty x_j(t)A_H(j)\,,\qquad
k\geq0\,.$$
The distribution of the quasispecies 
is the only positive stationary solution of the above system,
which exists for values of $a$
such that $\s\exa>1$,
and is given by the formula
$$\rho_k\,=\,(\s\exa-1)\frac{a^k}{k!}
\sum_{j=1}^\infty \frac{j^k}{\s^j}\,,\qquad
k\geq 0\,.$$
Our objective is to generalise this formula
to fitness functions $f:\N\lra[0,\infty[$
others than the sharp peak landscape fitness function.

\section{Results}
Let $f:\N\lra[0,\infty[$.
We consider the system of differential equations
$$x_k'(t)\,=\,
\sum_{j=0}^k x_j(t)f(j)\exa\frac{a^{k-j}}{(k-j)!}
-x_k(t)\sum_{j=0}^\infty x_j(t)f(j)\,,\qquad
k\geq0\,,$$
and we look for the stationary solutions of the system,
i.e., we want to solve the system of equations 
$$(\cS)\qquad
0\,=\,\sum_{j=0}^k x_jf(j)\exa\frac{a^{k-j}}{(k-j)!}
-x_k\Phi\,,\quad k\geq 0\,,$$
where $\Phi=\sum_{j\geq0}x_jf(j)$.
Since we think of $x_k$
as the concentration of the Hamming class $k$
in a population, we are only interested 
in non--negative solutions of the system $(\cS)$.
We say that $(x_k)_{k\geq0}$
is a quasispecies associated to $f$
if it is a non--negative solution of $(\cS)$
such that $x_0>0$ and $\sum_{k\geq0}x_k=1$.

\textbf{Assumption.}
We suppose that the 
fitness of the Hamming class $0$
is higher than the fitness of all the other classes,
i.e., the fitness function ${f:\N\lra[0,\infty[}$
satisfies $f(0)>f(k)$, $k\geq1$.

Note that the hypothesis is coherent
with the Hamming class $0$ corresponding
to the master sequence (the fittest genotype).
From now on, every fitness function is assumed
to verify this hypothesis.
We fix one such fitness function $f$ and
we focus ourselves on finding
the quasispecies distributions associated to $f$.

Let us remark that under this assumption,
if $(x_k)_{k\geq0}$ is a quasispecies,
then the concentration $x_k$ of the Hamming class $k$
is strictly positive. Indeed,
since we assume that $x_0>0$,
$$x_k\,=\,\frac{\displaystyle
\sum_{0\leq j\leq k}x_jf(j)\exa\frac{a^{k-j}}{(k-j)!}}
{\displaystyle\sum_{j\geq0}x_jf(j)}\,\geq\,
\frac{\displaystyle x_0f(0)\exa\frac{a^k}{k!}}{f(0)}\,>\,0\,.$$

The first of our results expresses the fitnesses 
as a function of the concentrations of the quasispecies.
\begin{theorem}\label{fitness}
Let us suppose that $(x_k)_{k\geq0}$ is a quasispecies
associated to $f$.
Then,
$$\forall k\geq 1\qquad
f(k)\,=\,\frac{f(0)}{x_k}
\sum_{j=0}^k (-1)^j\frac{a^j}{j!}x_{k-j}\,.$$
\end{theorem}
The interest of this result lies in its potential applications.
When performing practical experiments,
the concentrations of the different genotypes can be measured,
and one delicate question is to infer
the underlying fitness landscape.
Recent progresses allow even to sequence
in--vivo virus populations,
and the quasispecies model is one of the main tools
employed in order to infer the fitness landscape
from the experimental data~\cite{SGMGB}.  

We look now for an inverse formula,
in other words,
we want to express the concentrations
of the different Hamming classes
as a function of the fitnesses.
let $(x_k)_{k\geq0}$ be a quasispecies 
associated to $f$.
The equation for $k=0$ in the system $(\cS)$ is
$$0\,=\,x_0\big(
f(0)\exa-\Phi
\big)\,.$$
Since we suppose that $x_0>0$, we have $\Phi=f(0)\exa$.
Replacing $\Phi$ by $f(0)\exa$ in $(\cS)$
we obtain a recurrence relation for $(x_k)_{k\geq0}$.
To begin with,
we will try to solve the recurrence relation with 
initial condition equal to $1$, i.e.,
$$(\cR)\qquad\qquad
\begin{array}{l}
y_0\,=\,1\,,\\
\displaystyle y_k\,=\,
\frac{1}{f(0)-f(k)}\sum_{j=0}^{k-1}y_jf(j)\frac{a^{k-j}}{(k-j)!}\,,
\qquad k\geq1\,.
\end{array}
$$
\begin{lemma}\label{conv}
Let $(y_k)_{k\geq0}$
be the solution of the recurrence relation $(\cR)$.

\parskip 0 pt
$\bullet$
If the series associated to $(y_k)_{k\geq0}$ converges,
there exists a unique quasispecies
$(x_k)_{k\geq0}$ associated to $f$, which is given by:
$$x_0\,=\,\bigg(\sum_{k\geq0}y_k\bigg)^{-1}\,,\qquad
x_k\,=\,x_0y_k\,,\quad k\geq1\,.$$
$\bullet$ 
If the series associated to $(y_k)_{k\geq0}$ diverges,
no quasispecies associated to $f$ exists.
\end{lemma}
\begin{proof}
The first statement of the lemma is obviously true.
For the second one,
note that if $(x_k)_{k\geq0}$ 
is a quasispecies associated to $f$,
then the sequence
$(y_k)_{k\geq0}$ defined by
$y_k=x_k/x_0$, $k\geq0$,
satisfies the recurrence relation $(\cR)$,
and the series associated to $(y_k)_{k\geq0}$ converges.
\end{proof}
Next we give three different explicit formulas
for the sequence $(y_k)_{k\geq0}$.
The first of the formulas involves multinomial coefficients.
\begin{theorem}\label{multinom}
For all $k\geq1$,
$$y_k\,=\,\frac{a^k}{k!}\frac{f(0)}{f(0)-f(k)}\\
\Bigg(1+\!\!\!\!\!
\sum_{\genfrac{}{}{0pt}{1}
{1\leq h<k}
{1\leq i_1<\dots<i_h<k}}\!\!\!\!\!
\frac{k!}{i_1!(i_2-i_1)!\dots(k-i_h)!}
\prod_{t=1}^h\frac{f(i_t)}{f(0)-f(i_t)}
\Bigg)\,.$$
\end{theorem}
\textbf{Up--down coefficients.} 
The sequence $(y_k)_{k\geq0}$ 
can also be expressed in terms of up--down coefficients.
Let us first introduce the up--down numbers or coefficients~\cite{Shev}.
Let $n\geq2$ and let 
$(q_1,\dots,q_{n-1})\in\lbrace\,-1,1\,\rbrace^{n-1}$.
We say that a permutation
${\s=(\s(1),\dots,\s(n))}$ 
of $\lbrace\,1,\dots,n\,\rbrace$
has Niven's signature $(q_1,\dots,q_{n-1})$
if for every ${i\in\lbrace\,1,\dots,n-1\,\rbrace}$,
the product $q_i(\s(i+1)-\s(i))$ is positive~\cite{Niven}.
\begin{definition}
Let $n\geq2$, $0\leq h<n$
and $0=i_0<i_1<\cdots<i_h<n$.
The up--down coefficient
$$\updown{n}{i_0,\dots,i_h}$$
is defined as the number of permutations of
$\lbrace\,1,\dots,n\,\rbrace$
having Niven's signature $(q_1,\dots,q_{n-1})$ given by
$$\forall i\in\lbrace\,1,\dots,n-1\,\rbrace\qquad
q_i\,=\,\begin{cases}
\quad 1 &\quad\text{if}\ i\in\lbrace\,i_1,\dots,i_h\,\rbrace\,,\\
\quad -1 &\quad\text{otherwise}\,.
\end{cases}$$
\end{definition}
\begin{theorem}\label{updown}
For all $k\geq1$
$$
y_k\,=\,\frac{a^k}{k!}
\Bigg(\prod_{j=1}^k\frac{f(0)}{f(0)-f(j)}\Bigg)
\!\!\!\sum_{\genfrac{}{}{0pt}{1}{0\leq h<k}
{0=i_0<\dots<i_h<k}}\!\!\!\!\!\!\!
\updown{k}{i_0,\dots,i_h}\prod_{t=0}^h\frac{f(i_t)}{f(0)}\,.
$$
\end{theorem}
Our last result concerns fitness functions
that are eventually constant.
For such functions 
we can express the concentrations $(y_k)_{k\geq0}$
in terms of the concentrations $(q_k)_{k\geq0}$
of the quasispecies associated to the sharp peak landscape 
fitness function.
Let $f$ be a fitness function which is eventually constant
equal to $c>0$.
Define $(q_k)_{k\geq0}$
as the solution to the recurrence relation $(\cR)$
with fitness function $(f(0),c,c,\dots)$, i.e.,
$$q_k\,=\,
(f(0)/c-1)\frac{a^k}{k!}
\sum_{i\geq1}\frac{i^k}{(f(0)/c)^i}\,.$$
\begin{theorem}\label{quasi}
Let $N\geq0$ be such that
$$f(N)\,\neq\,c\qquad
\text{and}\qquad
\forall k>N
\quad
f(k)\,=\,c\,.$$
Then, for all $k>N$,
\begin{multline*}
y_k\,=\,q_k+
\sum_{j=1}^N\frac{a^j}{j!}q_{k-j}\bigg(
\frac{f(j)}{f(0)-f(j)}-\frac{c}{f(0)-c}
\bigg)\\
\times\bigg(
1+\sum_{h=1}^{j-1}\,\sum_{0=i_0<\cdots<i_h<j}\,
\prod_{t=1}^h\binom{j-i_{t-1}}{j-i_t}\frac{f(i_t)}{f(0)-f(i_t)}
\bigg)\,.
\end{multline*}
\end{theorem}
Finally we give a condition 
guaranteeing the existence of a quasispecies associated to $f$.
Let us recall that $(x_k)_{k\geq0}$
is a quasispecies associated to $f$ 
if it is a non negative solution of $(\cS)$,
$x_0>0$ and the sum of the $x_k$s is 1.
\begin{corollary}\label{serreur}
We have:

\parskip 0 pt
$\bullet$ If $f(0)\exa>\limsup_{n\to\infty}f(n)$,
the series associated to $(y_k)_{k\geq0}$ converges
and there exists a unique quasispecies associated to $f$.

$\bullet$ If $f(0)\exa<\liminf_{n\to\infty}f(n)$,
the series associated to $(y_k)_{k\geq0}$ diverges
and no quasispecies associated to $f$ exists.
\end{corollary}
We remark that for fitness functions $f$
(verifying our assumption),
the above corollary corresponds to 
the error threshold phenomenon observed by Eigen.
Moreover,
the error threshold depends only 
on $a$, $f(0)$,
and the limiting behaviour of the fitness function $f$.

To finish this section,
we discuss the motivations for making the assumption
${f(0)>f(k)}$, $k\geq1$,
and why we are mainly interested in solutions
$(x_k)_{k\geq0}$
of $(\cS)$
satisfying $x_0>0$.
Let $K\geq0$,
we call $(x_k)_{k\geq0}$
a \emph{quasispecies distribution around $K$ associated to $f$},
if it is a non negative solution of $(\cS)$ such that
$$x_0\,=\,\cdots\,=\,x_{K-1}\,=\,0\,<\,x_K
\qquad\text{and}\qquad
\sum_{k\geq K}x_k\,=\,1\,.$$
\begin{lemma}
Let $K\geq0$, and define the mapping
$g:\N\lra[0,\infty[$ by
$$\forall k\geq 0\qquad
g(k)\,=\,f(K+k)\,.$$
The sequence $(x_k)_{k\geq0}$
is a quasispecies distribution around $K$ associated to $f$
if and only if the sequence $(x_{K+i})_{i\geq0}$
is a quasispecies distribution around $0$ associated to $g$.
\end{lemma}
\begin{proof}
Let the sequence $(x_k)_{k\geq0}$ 
be a quasispecies distribution around $K$ associated to $f$.
Since $x_0=\cdots=x_{K-1}=0$,
for all $k\geq K$ we have
$$0\,=\,\sum_{j=K}^k x_jf(j)\exa\frac{a^{k-j}}{(k-j)!}
-x_k\sum_{j\geq K}x_jf(j)\,.$$
We set $i=k-K$ and $h=j-K$ in the above formula
and we see that for all $i\geq 0$,
\begin{align*}
0\,&=\,\sum_{h=0}^i x_{K+h}f(K+h)\exa\frac{a^{i-h}}{(i-h)!}
-x_{K+i}\sum_{h\geq0}x_{K+h}f(K+h)\\
&=\,\sum_{h=0}^i x_{K+h}g(h)\exa\frac{a^{i-h}}{(i-h)!}
-x_{K+i}\sum_{h\geq0}x_{K+h}g(h)\,.
\end{align*}
Therefore, the sequence
$(x_{K+i})_{i\geq0}$
is a quasispecies distribution around $0$ associated to $g$.
The converse implication is proved similarly.
\end{proof}
\begin{lemma}
Suppose there exists $K\geq1$ such that
$f(K)>\max_{0\leq k<K}f(k)$.
Then, for $k\in\lbrace\,0,\dots,K-1\,\rbrace$,
no quasispecies distribution around $k$ associated to $f$
exists.
\end{lemma}
\begin{proof}
Let us suppose that the sequence $(x_k)_{k\geq0}$
is a solution of $(\cS)$.
Let $k\in\lbrace\,0,\dots,K-1\,\rbrace$
and let us suppose further that $x_0=\cdots=x_{k-1}=0$
and $x_k\neq0$.
We will show that if $x_k>0,\dots,x_{K-1}>0$,
then necessarily $x_K<0$.
On one hand, 
writing down the $K$--th equation of $(\cS)$
we see that
$$x_K\,=\,\frac{1}{\Phi-f(K)\exa}
\sum_{j=k}^{K-1}x_jf(j)\exa\frac{a^{K-j}}{(K-j)!}\,.$$
On the other hand, 
writing down the $k$--th equation of $(\cS)$,
since $x_0=\cdots=x_{k-1}=0$ and $x_k>0$,
we conclude that $\Phi=f(k)\exa$.
Since $f(k)<f(K)$, 
if $x_k>0,\dots,x_{K-1}>0$,
necessarily $x_K<0$.
This implies that no quasispecies distribution
around $k$ associated to $f$ exists.
\end{proof}

The above lemmas justify the hypothesis on the fitness function
$f$, as well as the search for quasispecies distribution around
$0$ associated to $f$.
From now onwards, if $(x_k)_{k\geq0}$
is quasispecies distribution around $0$
associated to $f$,
and when there is no confusion,
we will simply say that $(x_k)_{k\geq0}$
is a quasispecies.

\section{Related results}
We have given three different explicit formulas
for the stationary solutions
of the system:
$$x_k'(t)\,=\,
\sum_{j=0}^k x_j(t)f(j)\exa\frac{a^{k-j}}{(k-j)!}
-x_k(t)\sum_{j=0}^\infty x_j(t)f(j)\,,\qquad
k\geq0\,,$$
As we have pointed out in the introduction,
this infinite system of differential equations
arises from Eigen's system of differential equations:
$$x_k'(t)\,=\,
\sum_{j=0}^\ell x_j(t)f(j)M_H(j,k)
-x_k(t)\sum_{j=0}^\ell x_j(t)f(j)\,,\qquad
0\leq k\leq \ell\,,$$
when considering the asymptotic regime
$$\ell\lra\infty\,,\qquad
q\lra0\,,\qquad
\ell q\lra a\in[0,\infty[\,\,.$$
Eigen's system of differential equations 
might be defined with greater generality:
given an at most countable set of types $\cG$,
a non negative fitness function $f$ on $\cG$,
and a stochastic matrix $M=\big(M(u,v),u,v\in\cG\big)$,
Eigen's model becomes
$$(*)\qquad
x_v'(t)\,=\,\sum_{u\in\cG}x_u(t)f(u)M(u,v)
-x_v(t)\sum_{u\in\cG}x_u(t)f(u)\,,\quad u\in\cG\,.$$
Define the matrix $W$ by setting
$$\forall u,v\in\cG\,,\qquad
W(u,v)\,=\,f(u)M(u,v)\,.$$
For a finite state space $\cG$
and under the hypothesis that the matrix $W$
is irreducible,
an application of the Perron--Frobenius
theorem for positive matrices 
shows that the system $(*)$
has a unique stationary solution
which is globally stable~\cite{TM,JER,Jones77,BK83}.
A similar result was proven by Moran~\cite{Moran76}
for a discrete--time version of this model:
$$(**)\qquad
x_v(n+1)\,=\,\frac{
\displaystyle\sum_{u\in\cG} x_u(n)f(u)M(u,v)}{
\displaystyle\sum_{u\in\cG} x_u(n)f(u)}\,,\quad
u\in\cG\,.$$
Once again, an application of the Perron--Frobenius
theorem shows that the dynamical system $(**)$
has a unique fixed point, which is globally stable.
Of course, the stationary solution
of the continuous dynamical system
and the fixed point of the discrete dynamical system
are the same.
Moran also extended this result~\cite{Moran76,Moran77}
to the case where $\cG=\Z$ and mutations 
only happen between nearest neighbours,
i.e., for $q\in]0,1/2[$ and $i\in\Z$,
the mutation matrix $M$ is 
defined by:
$$M(i,j)\,=\,\begin{cases}
\quad q &\quad\text{if}\ j=i\pm 1\,,\\
\quad 1-2q &\quad\text{if}\ j=i\,,\\
\quad 0 &\quad\text{otherwise}\,.\\
\end{cases}$$
Kingman~\cite{Kingman} further generalises Moran's result.
Let $\cG=\N$ and make the following assumptions:

$\bullet$ The fitness function is positive and bounded,
i.e., there exists a constant $C>0$ such that
$$\forall k\geq0\,,\qquad 0\,<\,f(k)\,<\,C\,.$$

$\bullet$ The mutation matrix $M$ is irreducible and aperiodic.

Let $\l$ be the spectral radius of the matrix $W$. 
Kingman then shows that if
$$\limsup_{k\to\infty}f(k)\,<\,\l\,,$$
then there exists a unique
positive
fixed point of $(**)$ having $1$
as the sum if its components.
Moreover, this fixed point is globally stable.
Kingman's result generalises the first statement
of our corollary~\ref{serreur}.
Indeed, in our setting $\s\exa$ corresponds to 
the spectral radius $\l$.
Our result, however, 
does not follow directly from Kingman's result,
for he assumes the matrix $W$ to be recurrent,
which is not verified in our case.
Kingman's proof,
which is based on an infinite dimensional version
of the Perron--Frobenius theorem,
could be extended to show 
the existence of a quasispecies, but not the uniqueness.
We have therefore chosen to exploit the obtained
explicit formulas  to derive an analogous 
of Kingman's result directly.
This procedure has not only allowed us
to retrieve Kingman's result
in our particular setting,
but also to give a similar condition
under which a quasispecies cannot be formed.

\section{Proof of theorems \ref{fitness}, 
\ref{multinom} and \ref{updown}}
\begin{proof}[Proof of theorem~\ref{fitness}]
Let us suppose that $(x_k)_{k\geq 0}$ 
is a quasispecies.
Let us show that, for all $k\geq 1$,
$$f(k)\,=\,\frac{f(0)}{x_k}
\sum_{j=0}^k (-1)^j\frac{a^j}{j!}x_{k-j}\,.$$
We will make the proof by induction.
The sequence $(x_k)_{k\geq0}$ is a quasispecies,
in particular, $x_0>0$ and $\Phi=f(0)\exa$.
Replacing $\Phi$ by $f(0)\exa$ in $(\cS)$
and arranging the terms gives
$$f(k)\,=\,\frac{1}{x_k}\Bigg(
x_kf(0)-\sum_{j=0}^{k-1}x_jf(j)\frac{a^{k-j}}{(k-j)!}
\Bigg)\,,\qquad k\geq1\,.$$
In particular, for $k=1$,
$$f(1)\,=\,\frac{f(0)}{x_1}(x_1-ax_0)\,.$$
So the result holds for $k=1$.
We fix now $k>1$ and we suppose
that the result holds up to $k-1$.
We replace the values of $f(1),\dots,f(k-1)$
in the above formula and we obtain
$$f(k)\,=\,\frac{f(0)}{x_k}\Bigg(
x_k-\sum_{j=0}^{k-1}\sum_{h=0}^j
(-1)^h x_{j-h}\frac{a^h}{h!}\frac{a^{k-j}}{(k-j)!}
\Bigg)\,.
$$
Let us fix $i\in\lbrace\,1,\dots,k\,\rbrace$
and let us look for the coefficient of $x_{k-i}$
in the above expression, this coefficient is
$$
-\sum_{\genfrac{}{}{0pt}{1}{0\leq h\leq j<k}{j-h=k-i}}
(-1)^h\frac{a^h}{h!}\frac{a^{k-j}}{(k-j)!}\,=\,
-\sum_{0\leq h<i}
(-1)^h\frac{a^h}{h!}\frac{a^{i-h}}{(i-h)!}\,
=\,(-1)^i\frac{a^i}{i!}\,,
$$
which concludes the proof of the theorem.
\end{proof}
\begin{proof}[Proof of theorem~\ref{multinom}]
We show that,
for all $k\geq 1$,
$$y_k\,=\,\frac{a^k}{k!}\frac{f(0)}{f(0)-f(k)}\\
\Bigg(1+\!\!\!\!\!
\sum_{\genfrac{}{}{0pt}{1}
{1\leq h<k}
{1\leq i_1<\dots<i_h<k}}\!\!\!\!\!
\frac{k!}{i_1!(i_2-i_1)!\dots(k-i_h)!}
\prod_{t=1}^h\frac{f(i_t)}{f(0)-f(i_t)}
\Bigg)\,.$$
Arranging the terms in $(\cR)$ gives:
$$y_k\big(f(0)-f(k)\big)\,=\,
\sum_{j=0}^{k-1}y_jf(j)\frac{a^{k-j}}{(k-j)!}\,,\qquad
k\geq 1\,.$$
We make the following changes of variables:
\begin{align*}
z_0\,&=\,y_0\,,\qquad
&&\phantom{=} g(0)\,=\,f(0)\,,\\
z_j\,&=\,\frac{y_j}{a^j}\big(f(0)-f(j)\big)\,,\qquad
&&\phantom{=} g(j)\,\,=\,\frac{f(j)}{f(0)-f(j)}\,,\qquad j\geq1\,.
\end{align*}
With these changes of variables,
the recurrence relation becomes
$$z_k\,=\,\sum_{j=0}^{k-1}z_j\frac{g(j)}{(k-j)!}\,,\qquad
k\geq1\,.$$
We iterate this formula and we obtain, for all $k\geq 1$,
$$z_k=z_0g(0)\Bigg(
\frac{1}{k!}+\sum_{h=1}^{k-1}\,\sum_{1\leq i_1<\cdots<i_h<k}
\frac{1}{i_1!(i_2-i_1)!\cdots(k-i_h)!}
\prod_{t=1}^h g(i_t)
\Bigg).$$
We replace $z_0,z_k$ and $g(0),\cdots,g(k-1)$
by their respective values and we obtain the desired result.
\end{proof}
\begin{proof}[Proof of theorem~\ref{updown}]
We show that, for all $k\geq1$,
$$
y_k\,=\,\frac{a^k}{k!}
\Bigg(\prod_{j=1}^k\frac{f(0)}{f(0)-f(j)}\Bigg)
\!\!\!\sum_{\genfrac{}{}{0pt}{1}{0\leq h<k}
{0=i_0<\dots<i_h<k}}\!\!\!\!\!\!\!
\updown{k}{i_0,\dots,i_h}\prod_{t=0}^h\frac{f(i_t)}{f(0)}\,.
$$
We take the formula from theorem~\ref{multinom} and
we set $\prod_{j=1}^k\big(f(0)-f(j)\big)$
as a common denominator, we get
\begin{multline*}
y_k\,=\,\frac{a^k}{k!}
\frac{f(0)}{\displaystyle\prod_{1\leq j\leq k}\big(f(0)-f(j)\big)}
\bigg(
\prod_{1\leq j<k}\big(f(0)-f(j)\big)\\
+\sum_{h=1}^{k-1}\sum_{1\leq i_1<\dots<i_h<k}
\frac{k!}{i_1!(i_2-i_1)!\dots(k-i_h)!}
\prod_{t=1}^h f(i_t)\!\!\!\!
\prod_{\genfrac{}{}{0pt}{1}{1\leq j<k}{j\neq i_1,\dots,i_h}}
\big(f(0)-f(j)\big)
\bigg)\,.
\end{multline*}
The expression in the large parenthesis
is an homogeneous polynomial
of degree ${k-1}$ in the variables $f(0),\dots,f(k-1)$.
For each ${h\in\lbrace\,1,\dots,k-1\,\rbrace}$
and ${1\leq i_1<\dots<i_h<k}$,
we get a monomial of the form
$f(0)^{k-1-h}f(i_1)\cdots f(i_h)$.
We calculate the coefficient of each of these monomials
and we conclude that
\begin{multline*}
y_k
=\frac{a^k}{k!}
\frac{f(0)}{\displaystyle\prod_{1\leq j\leq k}\big(f(0)-f(j)\big)}
\Bigg(f(0)^{k-1}
+\sum_{h=1}^{k-1}\sum_{1\leq i_1<\dots<i_h<k}f(0)^{k-1-h}
\prod_{t=1}^h f(i_t)\\
\times\bigg(
(-1)^h+\sum_{t=1}^h\sum_{1\leq j_1<\dots<j_t\leq h}(-1)^{h-t}
\frac{k!}{i_{j_1}!(i_{j_2}-i_{j_1})!\dots(k-i_{j_t})!}
\bigg)
\Bigg)\,.
\end{multline*}
We know from~\cite{Car} that
$$\updown{k}{0,i_1,\dots,i_h}\,=\,
(-1)^h+\sum_{t=1}^h\sum_{1\leq j_1<\dots<j_t\leq h}(-1)^{h-t}
\frac{k!}{i_{j_1}!(i_{j_2}-i_{j_1})!\dots(k-i_{j_t})!}\,,$$
which implies the desired result.
\end{proof}
\section{Proof of theorem~\ref{quasi}}
Let us introduce some notation before jumping
into the proof of the theorem.
For a fitness landscape $f$ and $k\geq 0$,
we define the fitness landscape $f^{(k)}$
obtained by shifting $k$
places to the left the fitnesses of the 
different classes and keeping the fitness
of the class $0$,
that is,
$$\forall j\geq0\qquad
f^{(k)}(j)\,=\,\begin{cases}
\quad f(0)&\quad\text{if}\ j=0\,,\\
\quad f(j+k)&\quad\text{if}\ j\geq1\,.
\end{cases}
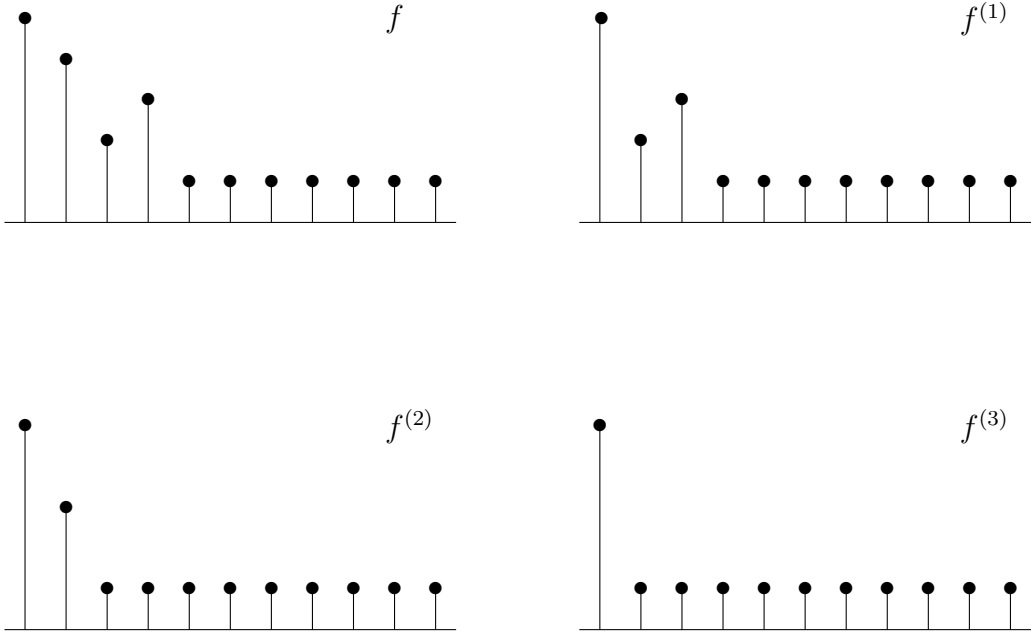
$$

\begin{tikzpicture}[scale=0.54]
\draw (0,0) -- (11,0) ;
\draw (0.5,0) -- (0.5,5) ;
\draw (1.5,0) -- (1.5,4) ;
\draw (2.5,0) -- (2.5,2) ;
\draw (3.5,0) -- (3.5,3) ;
\draw (4.5,0) -- (4.5,1) ;
\draw (5.5,0) -- (5.5,1) ;
\draw (6.5,0) -- (6.5,1) ;
\draw (7.5,0) -- (7.5,1) ;
\draw (8.5,0) -- (8.5,1) ;
\draw (9.5,0) -- (9.5,1) ;
\draw (10.5,0) -- (10.5,1) ;
\draw (9,5) node[right]{$f$} ;
\draw (0.5,5) node{$\bullet$} ;
\draw (1.5,4) node{$\bullet$} ;
\draw (2.5,2) node{$\bullet$} ;
\draw (3.5,3) node{$\bullet$} ;
\draw (4.5,1) node{$\bullet$} ;
\draw (5.5,1) node{$\bullet$} ;
\draw (6.5,1) node{$\bullet$} ;
\draw (7.5,1) node{$\bullet$} ;
\draw (8.5,1) node{$\bullet$} ;
\draw (9.5,1) node{$\bullet$} ;
\draw (10.5,1) node{$\bullet$} ;

\draw (14,0) -- (25,0) ;
\draw (14.5,0) -- (14.5,5) ;
\draw (15.5,0) -- (15.5,2) ;
\draw (16.5,0) -- (16.5,3) ;
\draw (17.5,0) -- (17.5,1) ;
\draw (18.5,0) -- (18.5,1) ;
\draw (19.5,0) -- (19.5,1) ;
\draw (20.5,0) -- (20.5,1) ;
\draw (21.5,0) -- (21.5,1) ;
\draw (22.5,0) -- (22.5,1) ;
\draw (23.5,0) -- (23.5,1) ;
\draw (24.5,0) -- (24.5,1) ;
\draw (23,5) node[right]{$f^{(1)}$} ;
\draw (14.54,5) node{$\bullet$} ;
\draw (15.5,2) node{$\bullet$} ;
\draw (16.5,3) node{$\bullet$} ;
\draw (17.5,1) node{$\bullet$} ;
\draw (18.5,1) node{$\bullet$} ;
\draw (19.5,1) node{$\bullet$} ;
\draw (20.5,1) node{$\bullet$} ;
\draw (21.5,1) node{$\bullet$} ;
\draw (22.5,1) node{$\bullet$} ;
\draw (23.5,1) node{$\bullet$} ;
\draw (24.5,1) node{$\bullet$} ;

\draw (0,-10) -- (11,-10) ;
\draw (0.5,-10) -- (0.5,-5) ;
\draw (1.5,-10) -- (1.5,-7) ;
\draw (2.5,-10) -- (2.5,-9) ;
\draw (3.5,-10) -- (3.5,-9) ;
\draw (4.5,-10) -- (4.5,-9) ;
\draw (5.5,-10) -- (5.5,-9) ;
\draw (6.5,-10) -- (6.5,-9) ;
\draw (7.5,-10) -- (7.5,-9) ;
\draw (8.5,-10) -- (8.5,-9) ;
\draw (9.5,-10) -- (9.5,-9) ;
\draw (10.5,-10) -- (10.5,-9) ;
\draw (9,-5) node[right]{$f^{(2)}$} ;
\draw (0.5,-5) node{$\bullet$} ;
\draw (1.5,-7) node{$\bullet$} ;
\draw (2.5,-9) node{$\bullet$} ;
\draw (3.5,-9) node{$\bullet$} ;
\draw (4.5,-9) node{$\bullet$} ;
\draw (5.5,-9) node{$\bullet$} ;
\draw (6.5,-9) node{$\bullet$} ;
\draw (7.5,-9) node{$\bullet$} ;
\draw (8.5,-9) node{$\bullet$} ;
\draw (9.5,-9) node{$\bullet$} ;
\draw (10.5,-9) node{$\bullet$} ;

\draw (14,-10) -- (25,-10) ;
\draw (14.5,-10) -- (14.5,-5) ;
\draw (15.5,-10) -- (15.5,-9) ;
\draw (16.5,-10) -- (16.5,-9) ;
\draw (17.5,-10) -- (17.5,-9) ;
\draw (18.5,-10) -- (18.5,-9) ;
\draw (19.5,-10) -- (19.5,-9) ;
\draw (20.5,-10) -- (20.5,-9) ;
\draw (21.5,-10) -- (21.5,-9) ;
\draw (22.5,-10) -- (22.5,-9) ;
\draw (23.5,-10) -- (23.5,-9) ;
\draw (24.5,-10) -- (24.5,-9) ;
\draw (23,-5) node[right]{$f^{(3)}$} ;
\draw (14.5,-5) node{$\bullet$} ;
\draw (15.5,-9) node{$\bullet$} ;
\draw (16.5,-9) node{$\bullet$} ;
\draw (17.5,-9) node{$\bullet$} ;
\draw (18.5,-9) node{$\bullet$} ;
\draw (19.5,-9) node{$\bullet$} ;
\draw (20.5,-9) node{$\bullet$} ;
\draw (21.5,-9) node{$\bullet$} ;
\draw (22.5,-9) node{$\bullet$} ;
\draw (23.5,-9) node{$\bullet$} ;
\draw (24.5,-9) node{$\bullet$} ;
\end{tikzpicture}

\vspace*{1 em}
For a fitness landscape $f$,
we denote by $(y_k(f))_{k\geq1}$
the solution to the recurrence $(\cR)$
corresponding to the fitness landscape $f$.
We start by establishing the following lemma,
which expresses the value of $y_k(f)$
as a function of $y_{k-1}(f^{(1)}),\dots,y_1(f^{(k-1)})$.
\begin{lemma}\label{shift}
For all $k\geq 2$, we have
$$y_k(f)\,=\,
\frac{a^k}{k!}\frac{f(0)}{f(0)-f(k)}+
\sum_{j=1}^{k-1}\frac{a^j}{j!}
\frac{f(j)}{f(0)-f(j)}
y_{k-j}(f^{(j)})
\,.$$
\end{lemma}
\begin{proof}
Consider the identity of theorem~\ref{multinom},
$$
y_k(f)=\frac{a^k}{k!}\frac{f(0)}{f(0)-f(k)}
\bigg(1+\!\!\!\!\!\!\!\!
\sum_{\genfrac{}{}{0pt}{1}{1\leq h<k}
{1\leq i_1<\cdots<i_h<k}}\!\!\!\!\!\!\!
\frac{k!}{i_1!(i_2-i_1)!\cdots(k-i_h)!}
\prod_{t=1}^h\frac{f(i_t)}{f(0)-f(i_t)}
\bigg),$$
and decompose the above sum according to the value of 
the first index:
\begin{multline*}
\sum_{\genfrac{}{}{0pt}{1}{1\leq h<k}
{1\leq i_1<\cdots<i_h<k}}\!\!\!\!
\frac{k!}{i_1!(i_2-i_1)!\dots(k-i_h)!}
\prod_{t=1}^h\frac{f(i_t)}{f(0)-f(i_t)}
=
\sum_{i_1=1}^{k-1}\frac{f(i_1)}{f(0)-f(i_1)}\\
\times
\bigg(
\frac{k!}{i_1!(k-i_1)!}+\!\!\!\!\!\!
\sum_{\genfrac{}{}{0pt}{1}{2\leq h\leq k-i_1}
{i_1<i_2<\cdots<i_h<k}}\!\!\!\!
\frac{k!}{i_1!(i_2-i_1)!
\cdots(k-i_h)!}
\prod_{t=2}^h\frac{f(i_t)}{f(0)-f(i_t)}
\bigg)\,.
\end{multline*}
We make the following changes of variables:
\begin{multline*}
\hfill
j\,=\,i_1\,,\qquad
h'\,=\,h-1\,,\qquad
t'=t-1\,,\hfill\\
\hfill
i'_1\,=\,i_2-i_1\,,\qquad
\dots\,,\qquad
i'_{h'}\,=\,i_h-i_1\,.
\hfill
\end{multline*}
Note that in particular we have
$f(i_{t'})=f(i'_{t'}+i_1)=f^{(j)}(i'_{t'})$.
The previous expression becomes:
\begin{multline*}
\sum_{j=1}^{k-1}\frac{k!}{j!(k-j)!}\frac{f(j)}{f(0)-f(j)}\\
\times\bigg(1+
\sum_{\genfrac{}{}{0pt}{1}{1\leq h'<k-j}
{1\leq i'_1<\cdots<i'_h<k-j}}
\frac{(k-j)!}{i'_1!(i'_2-i'_1)!
\cdots(k-j-i'_{h'})!}
\prod_{t'=1}^{h'}\frac{f^{(j)}(i'_{t'})}{f^{(j)}(0)-f^{(j)}(i'_{t'})}
\bigg)\,.
\end{multline*}
Since
$f(0)=f^{(j)}(0)$
and 
$f(k)=f^{(j)}(k-j)$
for all $j\in\lbrace\,1,\dots,k-1\,\rbrace$,
taking away the $'$ from the indexes, we see that
\begin{multline*}
y_k(f)\,=\,\frac{a^k}{k!}\frac{f(0)}{f(0)-f(k)}+
\sum_{j=1}^{k-1}\frac{a^j}{j!}\frac{f(j)}{f(0)-f(j)}
\times\frac{a^{k-j}}{(k-j)!}
\frac{f^{(j)}(0)}{f^{(j)}(0)-f^{(j)}(k-j)}\\
\times\bigg(1+
\sum_{\genfrac{}{}{0pt}{1}{1\leq h<k-j}
{1\leq i_1<\cdots<i_h<k-j}}
\frac{(k-j)!}{i_1!(i_2-i_1)!
\cdots(k-j-i_h)!}
\prod_{t=1}^h\frac{f^{(j)}(i_t)}{f^{(j)}(0)-f^{(j)}(i_t)}
\bigg)\,.
\end{multline*}
Yet, by theorem~\ref{multinom},
\begin{multline*}
y_{k-j}(f^{(j)})\,=\,\frac{a^{k-j}}{(k-j)!}
\frac{f^{(j)}(0)}{f^{(j)}(0)-f^{(j)}(k-j)}\times\\
\bigg(1+
\sum_{h=1}^{k-j}\sum_{1\leq i_1<\dots<i_h<k-j}
\frac{(k-j)!}{i_1!(i_2-i_1)!\cdots(k-j-i_h)!}
\prod_{t=1}^h\frac{f^{(j)}(i_t)}{f^{(j)}(0)-f^{(j)}(i_t)}
\bigg)\,.
\end{multline*}
We replace in the previous formula and we conclude.
\end{proof}
Let $f:\N\lra[0,\infty[$
be a fitness function which is eventually constant,
i.e., there exist $N\geq0$
and a positive constant $c$ such that
$$f(N)\,\neq\,c\qquad
\text{et}\qquad
\forall k>N\quad
f(k)\,=\,c\,.$$
Let $\big(y_k(f)\big)_{k\geq0}$ 
be the solution to the recurrence relation $(\cR)$
for the fitness function $f$. 
We want to show that, for all $k>N$,
\begin{multline*}
y_k(f)\,=\,
q_k+
\sum_{j=1}^N\frac{a^j}{j!}q_{k-j}\bigg(
\frac{f(j)}{f(0)-f(j)}-\frac{c}{f(0)-c}
\bigg)\\
\times\Bigg(
1+\sum_{h=1}^{j-1}\,\sum_{0=i_0<\cdots<i_h<j}\,
\prod_{t=1}^h\binom{j-i_{t-1}}{j-i_t}\frac{f(i_t)}{f(0)-f(i_t)}
\Bigg)\,,
\end{multline*}
where $(q_k)_{k\geq 0}$ 
is the solution to the relation of recurrence $(\cR)$ 
for the sharp peak fitness landscape $(f(0),c,c,\dots)$, i.e.,
$$q_k\,=\,\big(f(0)/c-1)
\frac{a^k}{k!}\sum_{i\geq 1}
\frac{i^k}{(f(0)/c)^i}\,,\qquad
k\geq1\,.$$
Before proceeding to the proof of the theorem~\ref{quasi},
we introduce the following notation in order
to simplify the expression of the formula we want to prove.
For a fitness function $f$ and $l\geq1$,
we set
$$C_l(f)\,=\,
1+\sum_{h=1}^{l-1}\sum_{0=i_0<i_1<\cdots<i_h<l}
\prod_{t=1}^h\binom{l-i_{t-1}}{l-i_t}
\frac{f(i_t)}{f(0)-f(i_t)}\,.$$
\begin{lemma}\label{coefs}
The coefficients $C_i$, $i\geq2$, 
satisfy the recurrence relation
$$C_i(f)\,=\,
1+\sum_{j=1}^{i-1}\binom{i}{j}
\frac{f(j)}{f(0)-f(j)}C_{i-j}(f^{(j)})\,.$$
\end{lemma}
\begin{proof}
Let $i\geq2$.
For $j\in\lbrace\,1,\dots,i-1\,\rbrace$,
$$C_{i-j}(f^{(j)})\,=\,
1+\sum_{h=1}^{i-j-1}\sum_{0=l_0<l_1<\cdots<l_h<i-j}
\prod_{t=1}^h\binom{i-j-l_{t-1}}{i-j-l_t}
\frac{f(j+l_t)}{f(0)-f(j+l_t)}\,.$$
We replace $C_{i-1}(f^{(1)}),\dots,C_1(f^{(i-1)})$
in the above formula,
and we change the indexes in the following way: 
$$h'\,=\,h+1\,,\qquad
j\,=\,i_1\,,\qquad
j+l_1\,=\,i_2\,,\qquad
\dots\,,\qquad
j+l_h\,=\,i_{h'}\,.$$
Exchanging the order of the sums 
gives the desired formula for $C_i(f)$.
\end{proof}
\begin{proof}[Proof of theorem~\ref{quasi}]
We show the result by induction on $N$.
Let us suppose first that $N=1$ and let $k\geq2$.
Then all the fitness functions
$f^{(j)}$, $j\geq1$, are
equal to the sharp peak landscape fitness function.
Applying lemma~\ref{shift} gives:
$$
y_k(f)\,=\,
\frac{a^k}{k!}\frac{f(0)}{f(0)-c}q_0
+a\frac{f(1)}{f(0)-f(1)}q_{k-1}
+\sum_{j=2}^{k-1}
\frac{a^j}{j!}\frac{c}{f(0)-c}q_{k-j}
\,.$$
Yet,
the sequence $(q_k)_{k\geq 0}$
satisfies the recurrence relation $(\cR)$
for the fitness function $(f(0),c,c,\dots)$,
i.e.,
$$\forall k\geq 1\qquad
q_k\,=\,\frac{1}{f(0)-c}\bigg(
\frac{a^k}{k!}f(0)+c\sum_{j=1}^{k-1}\frac{a^{j}}{j!}
q_{k-j}
\bigg)\,.$$
It follows that
$$y_k(f)\,=\,q_k+
a\bigg(\frac{f(1)}{f(0)-f(1)}-\frac{c}{f(0)-c}
\bigg)q_{k-1}\,.$$
The base case $N=1$ is thus settled.
Let now $N\geq2$ and
let us suppose that the result of theorem~\ref{quasi}
holds up to $N-1$.
Let $k>N$.
On one hand, for all $j\geq N$, 
the fitness function $f^{(j)}$ is equal to
the sharp peak landscape fitness function, 
therefore $y_{k-j}(f^{(j)})=q_{k-j}$
for all $j\in\lbrace\,1,\dots,N\,\rbrace$.
On the other hand, $f(N+1)=\cdots=f(k)=c$.
Thus, applying lemma~\ref{shift} gives
\begin{multline*}
y_k(f)\,=\,\frac{a^k}{k!}\frac{f(0)}{f(0)-f(k)}+
\sum_{j=1}^{k-1}\frac{a^j}{j!}
\frac{f(j)}{f(0)-f(j)}y_{k-j}(f^{(j)})\\=
\,
\frac{1}{f(0)-c}\bigg(
\frac{a^k}{k!}f(0)+c\sum_{j=1}^{k-1}\frac{a^j}{j!}q_{k-j}
\bigg)
-\frac{c}{f(0)-c}\sum_{j=1}^{N}\frac{a^j}{j!}q_{k-j}
\\+\sum_{j=1}^{N}\frac{a^j}{j!}
\frac{f(j)}{f(0)-f(j)}y_{k-j}(f^{(j)})\\
\,=q_k-\frac{c}{f(0)-c}\sum_{j=1}^{N}\frac{a^j}{j!}q_{k-j}
+\sum_{j=1}^{N}\frac{a^j}{j!}
\frac{f(j)}{f(0)-f(j)}y_{k-j}(f^{(j)})\,.
\end{multline*}

By the induction hypothesis, 
for all $j\in\lbrace\,1,\dots,N\,\rbrace$,
we have
\begin{align*}
y_{k-j}(f^{(j)})\,&=\,
q_{k-j}
+\sum_{l=1}^{N-j}\frac{a^l}{l!}q_{k-j-l}\bigg(
\frac{f^{(j)}(l)}{f^{(j)}(0)-f^{(j)}(l)}
-\frac{c}{f^{(j)}(0)-c}\bigg)C_l(f^{(j)})\\
&=\,q_{k-j}
+\sum_{l=1}^{N-j}\frac{a^l}{l!}q_{k-j-l}\bigg(
\frac{f(j+l)}{f(0)-f(j+l)}
-\frac{c}{f(0)-c}\bigg)C_l(f^{(j)})\,.
\end{align*}
We replace
$y_{k-N}(f^{(N)}),\dots,y_{k-1}(f^{(1)})$
in the formula for $y_k(f)$
and we obtain
\begin{multline*}
y_k(f)\,=\,q_k
-\frac{c}{f(0)-c}
\sum_{j=1}^N\frac{a^j}{j!}q_{k-j}
+\sum_{j=1}^N\frac{a^j}{j!}\frac{f(j)}{f(0)-f(j)}\\
\times\Bigg(
q_{k-j}+\sum_{l=1}^{N-j}\frac{a^l}{l!}q_{k-j-l}
\bigg(
\frac{f(j+l)}{f(0)-f(j+l)}-\frac{c}{f(0)-c}
\bigg)C_l\big(f^{(j)}\big)
\Bigg)\,.
\end{multline*}
Let us fix $i\in\lbrace\,1,\dots,N\,\rbrace$.
The coefficient of $q_{k-i}$
in the development of $y_k(f)$ is then equal to:
\begin{multline*}
\frac{a^i}{i!}\Bigg(\frac{f(i)}{f(0)-f(i)}
-\frac{c}{f(0)-c}\Bigg)
\\+\sum_{\genfrac{}{}{0pt}{1}
{1\leq j \leq N}{\genfrac{}{}{0pt}{1}{1\leq l\leq N-j}{j+l=i}}}
\frac{a^j}{j!}\frac{f(j)}{f(0)-f(j)}
\frac{a^l}{l!}\bigg(
\frac{f(j+l)}{f(0)-f(j+l)}-\frac{c}{f(0)-c}
\bigg)C_l(f^{(j)})\\
=\,\frac{a^i}{i!}\bigg(
\frac{f(i)}{f(0)-f(i)}-\frac{c}{f(0)-c}
\bigg)\Bigg(
1+\sum_{j=1}^{i-1}
\binom{i}{j}\frac{f(j)}{f(0)-f(j)}
C_{i-j}(f^{(j)})
\Bigg)\,.
\end{multline*}
We conclude thanks to lemma~\ref{coefs}.
\end{proof}

\section{Proof of the corollary}
We begin by giving two useful lemmas.
For a fitness function
${f:\N\lra[0,\infty[}$
we denote by $(y_k(f))_{k\geq0}$
the solution to the recurrence relation $(\cR)$
corresponding to the function $f$.
\begin{lemma}\label{compfitness}
Let
$f,g:\N\lra[0,\infty[$
be two fitness functions satisfying both
${f(0)=g(0)}$ and
$f(k)\geq g(k)$ for all $k\geq1$.
Then, for all $k\geq 0$,
we have $y_k(f)\geq y_k(g)$.
\end{lemma}
\begin{proof}
The result follows
from the inequality
$$\frac{1}{f(0)-f(k)}
\sum_{j=0}^{k-1}y_jf(j)\frac{a^{k-j}}{(k-j)!}\,\geq\,
\frac{1}{g(0)-g(k)}
\sum_{j=0}^{k-1}y_jg(j)\frac{a^{k-j}}{(k-j)!}\,,$$
along with an induction argument.
\end{proof}
Let $N\geq 1$ and $\s>c\geq0$.
We define the fitness function ${g_N:\N\lra[0,\infty[}$
by setting:
$$\forall k\geq0\,,\qquad
g_N(k)\,=\,
\begin{cases}
\quad \s\quad\text{if}\ k=0\,,\\
\quad 0\quad\text{if}\ 1\leq k\leq N\,,\\
\quad 1\quad\text{if}\ N+1\leq k\,.
\end{cases}$$
\begin{lemma}\label{serreursimple}
The series associated to $(y_k(g_N))_{k\geq0}$
converges if and only if $\s\exa>c$.
\end{lemma}
\begin{proof}
We know the result to be true
for the sharp peak landscape, 
i.e. for $N=0$.
By the comparison lemma~\ref{compfitness},
if $\s\exa>c$ 
the series associated to $(y_k(g_N))_{k\geq0}$
converges.
Suppose next that $\s\exa\leq c$.
By lemma~\ref{conv},
the convergence of the series associated to 
$(y_k(g_N))_{k\geq 0}$
is equivalent to the existence of 
a quasispecies associated to $g_N$.
We will thus show that such a quasispecies
cannot exist if $\s\exa\leq1$.
Let us suppose that 
a quasispecies $(x_k)_{k\geq0}$ exists.
The sequence $(x_k)_{k\geq0}$ then verifies:
\begin{align*}
x_0\,&>\,0\,,\quad
\sum_{k\geq0}x_k\,=\,1\,,\\
\Phi\,&=\,\s x_0+c\sum_{k>N}x_k\,,\\
0\,&=\,x_0(\s\exa-\Phi)\,,\\
0\,&=\,x_0\s\frac{a^k}{k!}\exa-x_k\Phi\,,\quad
1\leq k\leq N\,.
\end{align*}
In particualr $\Phi=\s\exa$ and $x_k=x_0 a^k/k!$
for $1\leq k\leq N$. Thus,
$$\s\exa\,=\,\Phi\,=\,
\s x_0 +c\big(1-(x_0+\dots+x_N)\big)\,=\,
\bigg(
\s-c\sum_{k=0}^N\frac{a^k}{k!}
\bigg)x_0+c\,.$$ 
Let
$$t_N(a)\,=\,\sum_{k=0}^N\frac{a^k}{k!}\,.$$
We conclude that $x_0$ is given by
$$x_0\,=\,\frac{\s\exa-c}{\s-ct_N(a)}\,.$$
Denote by $a^*$ the only positive solution
to the equation $ct_N(a)=\s$.
The expression obtained for $x_0$
is not positive for
$a\in[\ln(\s/c),a^*[$, so a quasispecies
cannot exist for $a$ in this interval.
If on the contrary $a\geq a^*$, we have
$$x_0+\cdots+x_N\,=\,t_N(a)x_0\,=\,
\frac{\s\exa t_N(a)-ct_N(a)}{\s-ct_N(a)}\,.$$
However, $t_N(a)<e^a$, 
which implies that this last expression
is strictly larger than 1.
Thus,
a quasispecies cannot exist
if $a\geq a^*$ either.
\end{proof}

We proceed now to the proof of the corollary.
After lemma~\ref{conv},
corollary~\ref{serreur}
will be settled if we manage to show
that the series associated to $(y_k)_{k\geq0}$
converges if $f(0)\exa>\limsup_{n\to\infty}f(n)$,
and diverges if $f(0)\exa<\liminf_{n\to\infty}f(n)$.
Let us start by showing the former.
Suppose first that the function 
is constant equal to $c\geq0$ from $N$ onwards.
We can thus apply theorem~\ref{quasi},
and obtain
$$\forall k>N\,,\qquad
y_k\,=\,q_k
+\sum_{j=1}^N\frac{a^j}{j!}q_{k-j}
\bigg(
\frac{f(j)}{f(0)-f(j)}-\frac{c}{f(0)-c}
\bigg)C_j(f)\,.$$
It follows that
$$\sum_{k\geq0}y_k=
\sum_{k=0}^N y_k
+\sum_{k>N}q_k
+\sum_{j=1}^N\frac{a^j}{j!}\bigg(\sum_{k>N}q_{k-j}\bigg)
\bigg(
\frac{f(j)}{f(0)-f(j)}-\frac{c}{f(0)-c}
\bigg)C_j(f)\,.
$$
Yet, the series associated to $(q_k)_{k\geq0}$
is convergent for $f(0)\exa>c$.
If the function $f:\N\lra[0,\infty[$
is not eventually constant, we set
$$c^\infty\,=\,\limsup_{n\to\infty}f(n)\,.$$
Let $\e>0$, pick $N\geq0$ large enough so that for all
$k>N$, $f(k)<c^\infty+\e$.
We define the function $f^N:\N\lra[0,\infty[$ by:
$$
\forall k\geq 0\qquad
f^N(k)\,=\,\begin{cases}
\quad f(k)\quad\text{if}\ 0\leq k\leq N\,,\\
\quad c^\infty+\e\quad\text{if}\ k>N\,.
\end{cases}
$$
For $\e$ small enough,
$f^N(0)\exa>c^\infty+\e$.
Since $f^N$ is constant equal to $c^\infty+\e$ 
from $N$ onwards,
the series associated to $(y_k(f^N))_{k\geq0}$ converges.
By the comparison lemma~\ref{compfitness}, 
the same is true for the series associated to
$(y_k(f))_{k\geq0}$.
We prove next that if
${f(0)\exa<\liminf_{n\to\infty}f(n)}$,
then the series associated to $(y_k)_{k\geq0}$ diverges.
We define the function $f_N:\N\lra[0,\infty[$ by:
$$
\forall k\geq 0\qquad
f_N(k)\,=\,\begin{cases}
\quad f(0)\quad\text{if}\ k=0\,,\\
\quad 0\quad\text{if}\ 1\leq k\leq N\,,\\
\quad c_\infty-\e\quad\text{if}\ k>N\,.
\end{cases}
$$
For $\e$ small enough,
$f^N(0)\exa<c^\infty-\e$.
After lemma~\ref{serreursimple}
the series associated to $(y_k(f^N))_{k\geq0}$ 
is divergent.
By the comparison lemma~\ref{compfitness}, 
the same is true for the series associated to
$(y_k(f))_{k\geq0}$.

\section{Conclusions}
We have given several explicit formulas for the
stationary solutions of Eigen's quasispecies model
in the regime where the length of the genotypes 
goes to infinity.
Theorem~\ref{fitness} allows the inference 
of the fitness landscape from data about the 
concentrations of the different genotypes,
which makes it particularly attractive for applications.
The formulas in theorems~\ref{multinom} and~\ref{updown}
give the concentrations of the different Hamming classes,
relative to the master sequence.
For fitness landscapes which are eventually constant,
a link is made to the already known distribution of
the quasispecies (for the sharp peak landscape~\cite{CD})
in theorem~\ref{quasi}.
Finally corollary~\ref{serreur} generalises 
the error threshold criterion observed for the sharp peak landscape
to fitness landscapes depending on the Hamming class.
The main interest of our results lies in their exact
nature; the only approximation they rely on 
is the long chain regime,
which even the simplest genomes in nature fall into.
The main limitation of our work is the assumption that
the fitness of an individual depends on its genome
only through the number of point mutations 
from the master sequence.
Nevertheless, we believe that our results
provide a first step in finding quasispecies distributions
for even more general fitness landscapes.

\bibliography{gen}
\bibliographystyle{plain}
\end{document}